\ifx\tics\undefined
    \documentclass[sigconf,nonacm]{acmart}
    \setcopyright{none} % No copyright notice required for submissions
    \acmConference{Crypto Economics Security Conference}{Oct 31 - Nov 1, 2022}{Berkeley, CA, USA}
    \acmISBN{}
    \acmDOI{}
\else
    \documentclass[a4paper,UKenglish,cleveref, autoref, thm-restate]{oasics-v2021}
\fi

\usepackage{graphicx}
\usepackage{url}
\usepackage{color}    % Used to highlight comments
\usepackage[ruled]{algorithm2e}
\usepackage{makecell}
\ifx\tics 
    \usepackage{mathptmx}
\fi
\usepackage{paralist}

\newcommand{\viewA}{viewA\xspace}
\newcommand{\viewB}{viewB\xspace}

\newcommand{\fino}[1]{{\tt #1}}

\newcommand{\BFTlong}{Byzantine fault tolerance\xspace}
\newcommand{\BFT}{BFT\xspace} 
\newcommand{\OTR}{commit-reveal\xspace}
\newcommand{\Disperse}[1]{Disperse(#1)\xspace}
\newcommand{\Retrieve}[1]{Retrieve(#1)\xspace}

\ifx\tics\undefined
    
\begin{document}
    \title{Maximal Extractable Value (MEV) Protection on a DAG}
    \author{Dahlia Malkhi and Pawel Szalachowski\\
        \small{Chainlink Labs}\\
    }

\else
    \title{Presentation and publication: Maximal Extractable Value (MEV) Protection on a DAG}
    \bibliographystyle{plainurl}% the mandatory bibstyle
    %\titlerunning{Dummy short title} %TODO optional, please use if title is longer than one line
    \author{Dahlia Malkhi}{Chainlink Labs}{}{}{}%TODO mandatory, please use full name; only 1 author per \author macro; first two parameters are mandatory, other parameters can be empty. Please provide at least the name of the affiliation and the country. The full address is optional. Use additional curly braces to indicate the correct name splitting when the last name consists of multiple name parts.
    \author{Pawel Szalachowski}{Chainlink Labs}{}{}{}
    \authorrunning{D. Malkhi and P. Szalachowski} %TODO mandatory. First: Use abbreviated first/middle names. Second (only in severe cases): Use first author plus 'et al.'
    \Copyright{Dahlia Malkhi and Pawel Szalachowski} %TODO mandatory, please use full first names. LIPIcs license is "CC-BY";  http://creativecommons.org/licenses/by/3.0/
     %TODO mandatory; please add comma-separated list of keywords
    \category{Publication and presentation} %optional, e.g. invited paper
    \relatedversion{} %optional, e.g. full version hosted on arXiv, HAL, or other respository/website
    %\relatedversiondetails[linktext={opt. text shown instead of the URL}, cite=DBLP:books/mk/GrayR93]{Classification (e.g. Full Version, Extended Version, Previous Version}{URL to related version} %linktext and cite are optional
    %\supplement{}%optional, e.g. related research data, source code, ... hosted on a repository like zenodo, figshare, GitHub, ...
    %\supplementdetails[linktext={opt. text shown instead of the URL}, cite=DBLP:books/mk/GrayR93, subcategory={Description, Subcategory}, swhid={Software Heritage Identifier}]{General Classification (e.g. Software, Dataset, Model, ...)}{URL to related version} %linktext, cite, and subcategory are optional
    %\funding{(Optional) general funding statement \dots}%optional, to capture a funding statement, which applies to all authors. Please enter author specific funding statements as fifth argument of the \author macro.
    \acknowledgements{We are grateful to Soumya Basu, Christian Cachin, Ari Juels, Mahimna Kelkar, Lefteris Kokoris-Kogias, Oded Naor, Mike Reiter for many comments that helped improve this writeup.}%optional
    %\nolinenumbers %uncomment to disable line numbering
    %Editor-only macros:: begin (do not touch as author)%%%%%%%%%%%%%%%%%%%%%%%%%%%%%%%%%%
    \EventEditors{John Q. Open and Joan R. Access}
    \EventNoEds{2}
    \EventLongTitle{4th International Conference on Blockchain Economics Security
    and Protocols}
    \EventShortTitle{Tokenomics 2022}
    \EventAcronym{Tokenomics}
    \EventYear{2022}
    \EventDate{December 13--14, 2022}
    \EventLocation{Paris, France}
    \EventLogo{}
    \SeriesVolume{42}
    \ArticleNo{23}
    %%%%%%%%%%%%%%%%%%%%%%%%%%%%%%%%%%%%%%%%%%%%%%%%%%%%%%
    %--------------------------------------------------------------------------------
    \begin{document}
\fi

\begin{abstract}
    Many cryptocurrency platforms are vulnerable to Maximal Extractable Value (MEV) attacks~\cite{daian2020flash}, where a malicious consensus leader can inject transactions or change the order of user transactions to maximize its profit.  

A promising line of research in MEV mitigation is to enhance the \BFTlong (\BFT) consensus core of blockchains by new functionalities, 
like hiding transaction contents, such that malicious parties cannot analyze and exploit them until they are ordered. 
    An orthogonal line of research demonstrates excellent performance for \BFT
    protocols designed around Directed Acyclic Graphs (DAG). They provide high
    throughput by keeping high network utilization, decoupling transactions'
    dissemination from their metadata ordering, and encoding consensus logic
    efficiently over a DAG representing a causal ordering of disseminated
    messages.

This paper explains how to combine these two advances.  It introduces a DAG-based protocol called Fino, that integrates MEV-resistance features into DAG-based \BFT without delaying the steady spreading of transactions by the DAG transport and with zero message overhead.
The scheme operates without complex secret share verifiability or recoverability, and avoids costly threshold encryption.

\end{abstract}

%--------------------------------------------------------------------------------

% TODO: replace this section with code generated by the tool at https://dl.acm.org/ccs.cfm
%\begin{CCSXML}
%<ccs2012>
%<concept>
%<concept_id>10002978.10003029.10011703</concept_id>
%<concept_desc>Security and privacy~Usability in security and privacy</concept_desc>
%<concept_significance>500</concept_significance>
%</concept>
%</ccs2012>
%\end{CCSXML}

%\ccsdesc{Security and privacy~Use https://dl.acm.org/ccs.cfm to generate actual concepts section for your paper}
% -- end of section to replace with generated code

% \keywords{template; formatting; pickling} % TODO: replace with your keywords

\maketitle

\section{Introduction}
\label{sec:intro}
% Many cryptocurrency platforms are vulnerable to blockchain extractable value (MEV) attacks~\cite{daian2020flash}, 
% where a malicious consensus leader can inject transactions or change the order of user transactions to maximize its profit.

% This paper presents another advantage of the DAG-based approach to BFT,
% enabling simple and smooth prevention of MEV exploits.
% It introduces a DAG-based protocol called Fino, 
% that leverages the DAG structure to achieve efficient MEV mitigation through {\em Blind Order-Fairness}.
% Integrating two additional lines of defense in a DAG is briefly discussed, Time-Based Order-Fairness and Participation Fairness.

% -----*-------*---------*--------*---------------
\subsection{MEV}
Over the last few years, we have seen exploding interest in cryptocurrency
platforms and applications built upon them, like decentralized finance protocols
offering censorship-resistant and open access to financial instruments; or
non-fungible tokens.
Many of these systems are vulnerable to MEV attacks, where a malicious consensus
leader can inject transactions or change the order of user transactions to maximize its
profit. Thus it is not surprising that at the same time we have witnessed
rising phenomena of MEV professionalization, where an entire ecosystem of MEV
exploitation, comprising of MEV opportunity "searchers" and collaborating
miners, has arisen.

Daian et al.~\cite{daian2020flash} introduced a measure
of the ``profit that can be made through including, excluding, or re-ordering transactions within blocks''.
The original work called the measure miner extractable value, which was later
extended by maximal extractable value (MEV)~\cite{obadia2021unity} and blockchain extractable value (BEV)~\cite{qin2021quantifying}, 
to include other forms of attacks, not necessarily performed by miners.
At the time of this writing, an ``MEV-explore''
tool~\cite{mev-explore}
estimates the amount of MEV extracted on Ethereum since the 1st of Jan 2020 to be close to \$700M. However,
it is safe to assume that the total MEV extracted is much higher, since
MEV-explore limits its estimates to only one blockchain, a few protocols, and a
limited number of detectable MEV techniques.  Although it is difficult to argue
that all MEV is "bad" (e.g., market arbitrage can remove market inefficiencies),
it usually introduces some negative externalities like:
%
% \begin{itemize}
% \item 
        \textit{network congestion}: especially on low-cost chains, MEV actors often try to
  increase their chances of exploiting an MEV opportunity by sending a lot of
  redundant transactions, spamming the underlying peer-to-peer network;
% \item 
        \textit{chain congestion}: many such transactions finally make it to the chain, making
  the chain more congested;
% \item 
        \textit{higher blockchain costs}: while competing for profitable MEV opportunities, MEV
  actors bid higher gas prices to prioritize their transactions, which results in
  overall higher blockchain costs for regular users;
% \item 
        \textit{consensus stability}: some on-chain transactions can create such a lucrative MEV
  opportunity that it may be tempting for miner(s) to create an alternative
  chain fork with such a transaction extracted by them, which
  introduces consensus instability risks.
% \end{itemize}

% -----*-------*---------*--------*---------------
\subsection{Blind Order-Fairness on a DAG}
\label{sec:intro:bof-dag}

A promising way of thwarting some forms of MEV on blockchains is to extend their \BFT (\BFTlong) consensus core 
by new properties like Blind Order-Fairness, where no consensus "validator" can 
learn the content of transactions until they are committed to a total ordering.
This notion of \emph{``\OTR''} is the key mechanism we focus on in this paper.

An orthogonal line of research concentrates on scaling the \BFT consensus 
of blockchains via a {\em DAG communication substrate}.
A DAG transport can spread yet-unconfirmed transactions in parallel, every message having utility and carrying transactions that eventually become committed to a total ordering. 
The DAG provides Reliability, Non-equivocation, and Causal-ordering delivery guarantees, 
enabling a simple and "zero-overhead" consensus ordering on top of it. 
That is, validators can interpret their DAG locally without exchanging more messages and determine a total ordering of accumulated transactions.
Several recent DAG-based \BFT systems, 
including Blockmania~\cite{danezis2018blockmania}, Aleph~\cite{gkagol2018aleph}, Narwhal~\cite{danezis2022narwhal},
DAG-Rider~\cite{keidar2021all}, and Bullshark~\cite{giridharan2022bullshark},
demonstrate excellent performance.

This paper combines these two advances in a solution called Fino, that integrates smoothly two ingredients: 
A framework for \OTR consensus protocols, and a simple DAG-based BFT consensus that 
leverages the Reliability, Non-Equivocation, and Causality properties of the DAG transport to integrate the \OTR framework smoothly and efficiently into consensus.

\textbf{A \OTR framework.} The core of the \OTR framework of Fino is a standard k-of-n secret sharing approach
consisting of two functionalities, \Disperse{} and \Retrieve{}. 
Users first send transactions to validators {\em encrypted}, 
so that the consensus protocol commits to an ordering of transactions {\em blindly}.
\Disperse{} entrusts shares of the secret decryption key for each transaction to validators.
After an ordering is committed, \Retrieve{} opens encrypted transactions by collecting
shares from $F+1$ validators.

Implementing \Disperse{}/\Retrieve{} efficiently enough to meet DAG-based consensus throughput is the core challenge Fino addresses.
Two straw man approaches are (i) threshold cryptography, but it would incur an order of milliseconds per transaction, (ii) using polynomial secret-sharing and verifying shares during dispersal, 
also requiring costly cryptography as well as implementing a share-recovery protocol.

The third approach, and the one we recommend using,
shares transaction encryption keys using Shamir's secret-sharing scheme~\cite{shamir1979share}, but
without secret share verifiability or recoverability (slow and complex), 
and without threshold encryption (slow).
Instead, 
to guarantee a unique, deterministic outcome of encrypted transactions,
Fino borrows a key insight from DispersedLedger~\cite{yang2022dispersedledger} called AVID-M, 
protecting against share manipulation attacks by verifying that the sharing was correct {\em post reconstruction}.
If the dispersal was incorrect and fails post-reconstruction verification, 
the transaction is rejected.

Notably, this approach works in microseconds latency and is orders of magnitude faster than threshold encryption.

\autoref{sec:framework} describes a fourth approach, a hybrid combining AVID-M with threshold cryptography. Hybrid tackles adversarial conditions where \Disperse{} partially fails, at the cost of a slower post-reconstruction verification.

We repeat that all four approaches above are compatible with the Fino framework, though the fastest and simplest one is AVID-M. 

\textbf{DAG-based integration.} Fino integrates \Disperse{}/\Retrieve{} into a DAG-based consensus 
without delaying the steady spreading of transactions by the DAG transport and with zero message overhead.  
Validators interpret their local DAGs to arrive at commit decisions on blind ordering. 
After an ordering has committed for a batch of transactions, 
validators invoke \Retrieve{} and present decryption shares piggybacked on DAG broadcasts. 
Finally, validators interpret their local DAGs to arrive at unanimous, deterministic outcome of transaction opening.

During periods of synchrony, Fino has a happy path for transaction commit of two DAG latencies, plus a single DAG latency for determining an opening outcome.

\section{Background and Preliminaries}
\label{sec:pre}
% -----*-------*---------*--------*---------------
\subsection{System Model} %BFT SMR problem, partial synchrony model}

The \BFTlong (\BFT) Consensus core of blockchains implements state machine replication (SMR),
where non-faulty parties--also known as validators--agree on an order of transactions and execute them consistently. 
The goal of our work is to extend \BFT consensus and build into it
protection against MEV attacks which rely on transaction content analysis.

We assume that the system should be resilient to
Byzantine faults, i.e., faulty validators can implement any adversarial behavior
(like crashing, equivocating, or colluding).
We assume a partially synchronous network~\cite{dwork1988consensus}, where asynchronous periods
last up to unknown global stabilization time (GST), while in synchronous periods
there is a known bound $\Delta$ on delays for message delivery.
This model requires a \BFT threshold $N \geq 3F + 1$, where
$N$ and $F$ is the number of all and Byzantine faulty validators, respectively.

% -----*-------*---------*--------*---------------
\subsection{MEV and Blind Order-Fairness}
\label{sec:pre:mev}
% \psz{most of that is in Sec 1.1, shall we move it here or drop this section?}
In this paper, we focus on consensus-level MEV mitigation techniques.
There are fundamentally two types of MEV-resistant Order-Fairness properties:

{\bf Blind Order-Fairness.}
A principal line of defense against MEV
stems from committing to transaction ordering without seeing transaction contents.
This notion of MEV resistance, referred to here as Blind Order-Fairness,
is used in a recent SoK on Preventing Transaction Reordering by Heimbach and
Wattenhofer~\cite{heimbach2022sok},
and is defined as:
{\em
"when it is not possible for any party to include or exclude transactions
after seeing their contents. Further, it should not be possible for any party
to insert their own transaction before any transaction whose contents it
already been observed."
}

{\bf Time-Based Order-Fairness.}
Another measure for MEV protection is brought by sending transactions to all \BFT parties simultaneously and using
the relative arrival order at a majority of the parties to determine the final ordering.
In particular, this notion of order fairness ensures that
{\em
"if sufficiently many parties receive a
transaction tx before another tx', then in the final commit order tx' is not sequenced before tx."
}

This defines the transactions order and  prevents adversaries that can analyze network traffic and
transaction contents from reordering, censoring, and front-/back-running
transactions received by Consensus parties. Moreover, Time-Based Order-Fairness
protects against a potential collusion between users and \BFT leaders/parties because parties explicitly input relative ordering into the protocol.
Time-Based Order-Fairness is used in various flavors in several recent works, including
Pompe~\cite{zhang2020byzantine},
Aequitas~\cite{kelkar2020order},
Themis~\cite{kelkar2021themis},
"Wendy Grows Up"~\cite{kursawe2021wendy}, and
"Quick Order Fairness"~\cite{cachin2021quick}.
We briefly discuss some of those protocols in \autoref{sec:related}.

Another notion of fairness found in the literature, that does not provide Order-Fairness, 
revolves around participation fairness:

{\bf Participation Fairness.}
A different notion of fairness aims to ensure 
censorship-resistance or stronger notions of participation equity.
Participation Fairness guarantees that the committed sequence of transactions includes a certain portion of honest contribution (aka "Chain Quality").
Several \BFT protocols address Participation Fairness, including
Prime~\cite{amir2010prime},
Fairledger~\cite{lev2019fairledger},
HoneyBadger~\cite{miller2016honey}.
As mentioned in \autoref{sec:dag}, some DAG-based BFT protocols like 
Aleph~\cite{gkagol2018aleph}, 
DAG-Rider~\cite{keidar2021all},
Tusk~\cite{danezis2022narwhal},
and Bullshark~\cite{giridharan2022bullshark} 
use a layered DAG paradigm.
In this approach, Participation Fairness comes essentially for free because every DAG layer must include messages from 2F+1 participants.
It is worth noting that Participation Fairness does not prevent 
a corrupt party from injecting transactions after it has already observed other transactions,
nor a corrupt leader from reordering transactions after reading them,
violating both Blind and Time-Based Order-Fairness.

% -----*-------*---------*--------*---------------
\subsection{Threshold Methods}
Beside standard cryptographic primitives, like a symmetric encryption scheme or
a cryptographic hash function $H()$, 
we employ primitives specifically designed for hiding and opening transactions.

{\bf Secret Sharing}
For positive integers $k,n$, where $k\leq n$, a \textit{(k,n)-threshold scheme}
is a technique where a secret is distributed among $n$ parties , such that it
can be reconstructed by any $k$ parties, while any group of $k-1$ parties
cannot learn the secret.  The secret $S$ is selected by a trusted third party, called
\textit{dealer}, who splits it into $n$ individual \textit{shares} by calling:
$s_1, ..., s_n\leftarrow \fino{SS.Split(S)}$.
After at least $k$ parties exchange shares, they can combine them and recover the
secret by calling $S\leftarrow$~\fino{SS.Combine($s'_1, ..., s'_k$)}.

Shamir's Secret Sharing (SSS)~\cite{shamir1979share} is one of the first and most widely used 
\textit{(k,n)}-threshold schemes.  In SSS,
\fino{Split($S$)} chooses a polynomial $f$ of degree $k-1$ that hides $S$ at the origin point, and generates $n$ points on the polynomial.
Then, any $k$ points allow to interpolate the polynomial and compute its value at $f(0)$.

% -----*-------*---------*--------*---------------
{\bf Threshold Encryption.}
A \textit{(k,n)-threshold encryption} is a scheme where messages
encrypted under a single public encryption key $pk$ can be decrypted by a private
decryption key, shared among $n$ parties -- each with its secret key $sk_i$, where any group of $k$ parties can
decrypt a message.
The message $m$ is encrypted with an encryption algorithm:
$c\leftarrow$~\fino{TE.Enc($pk, m$)}.
We consider schemes where each party, upon receiving a ciphertext $c$, computes
its own \textit{decryption share} by invoking:
$ds_i\leftarrow$~\fino{TE.ShareGen($sk_i, c$)}; and any set of $k$ unique decryption
shares is enough to recover the plaintext, which we denote 
$m\leftarrow$\fino{TE.Dec($c, ds'_1, ..., ds'_k$)}.

Moreover, some threshold encryption schemes, like shown by Shoup and
Gennaro~\cite{shoup1998securing}, allow parties to verify that a given
decryption share corresponds to the ciphertext, which we denote by the \fino{TE.Verify($c, ds_i$)} function.  Note that, without such a
verification, it is trivial for a malicious party, providing an incorrect share,
to cause decryption failures.

\section{The \OTR Framework}
\label{sec:framework}
The core of blind fairness-ordering is a standard k-of-n secret sharing approach. 
For each transaction \fino{tx}, a user (“Alice”) picks a secret symmetric key \fino{tx-key}, and sends \fino{tx} encrypted with it to validators.
Secret sharing allows Alice to share \fino{tx-key} with validators such that $F+1$ shares are required to reconstruct \fino{tx-key}, and no set of $F$ bad validators can open it before it is committed (blindly) to the total order. 
Honest validators wait to commit to a blind order of transactions first, and only later open them. At the same time, before committing a transaction to the total ordering, validators ensure that they can open it.

More formally, a sharing protocol has two abstract functionalities, \Disperse{tx} and \Retrieve{tx}. \Disperse{} allows a dealer to disseminate shares of the secret \fino{tx-key} to validators, with the following guarantees:

\begin{description}

\item[Hiding:] A set of F bad validators cannot reconstruct \fino{tx-key}.
\item[Binding:] After \Disperse{tx-key} completes successfully at an honest validator, any set of honest $F+1$ validators doing \Retrieve{tx-key} can reconstruct \fino{tx-key} generating a unique outcome.
\item[Validity:] If the dealer is honest, then the outcome from \Retrieve{tx-key} by any $F+1$ honest validators is \fino{tx-key}.

\end{description}

It is worth noting that \Disperse{}/\Retrieve{} does not require certain properties which some use-cases in the literature may need, but not the framework here:

\begin{itemize}
\item When \Disperse{} completes, it does not guarantee that \Retrieve{} can reconstruct an output value that a dealer commits to. 
\item It doesn’t require that individual shares can be recovered.
\item It doesn’t require being able to use \fino{tx-key} to derive other values while keeping \fino{tx-key} itself secret.
\end{itemize}

We proceed to present four approaches for implementing \Disperse{}/\Retrieve{}. All of them are compatible with the DAG-based consensus protocol in the next section, though the fastest and simplest one, which we recommend using, is AVID-M. 

% --- * ---- * ---- * ----
\subsection*{Approach-1: Threshold Cryptography}

It is straight-forward to implement Disperse/Retrieve using threshold encryption, such that the public encryption \fino{TE.Enc()} is known to users and the private decryption \fino{TE.Dec()} is shared (at setup time) among validators. 
To \Disperse{tx-key}, the user attaches \fino{TE.Enc(tx-key)}, the transaction key encrypted with the global threshold key. 
Once a \fino{tx} is committed to the total-ordering, \Retrieve{tx-key} is implemented by each validator generating its decryption share via \fino{TE.ShareGen()}. 
Validators can try applying \fino{TE.Dec()} over $F+1$ valid decryption shares to decrypt \fino{tx}. 

Binding stems from two properties. First, a threshold of honest validators can always succeed in generating $F+1$ valid decryption shares. 
Second, as indicated in \autoref{sec:pre}, some threshold cryptography schemes allow verifying that a validator is contributing a correct decryption share, hence by retrieving $F+1$ valid threshold shares, a unique outcome is guaranteed.

Unfortunately, this method is computationally somewhat heavy: \Disperse{} takes about $300 \mu s$ to encrypt using TDH2 on standard hardware, and
\Retrieve{} takes about $3500 \mu s$ for in a 6-out-of-16 scheme (see \autoref{sec:implementation}).

% --- * ---- * ---- * ----
\subsection*{Approach-2: VSS}

Another way for users to \Disperse{tx-key} is Shamir’s secret sharing (SSS) scheme. 
To \Disperse{tx-key}, a user employs \fino{SS.Split(tx-key)} to send individual shares of \fino{tx-key} to each validator. 
A set of $F+1$ share holders reveal their shares during \Retrieve{} and combine
shares via \fino{SS.Combine()} to reconstruct \fino{tx-key}.

Combining shares is three orders of magnitude faster than threshold crypto and takes a few microseconds in today’s computing environment (see \autoref{sec:implementation}).
However, after \Disperse{} completes with $N-F$ validators, “vanilla”
\fino{SS.Split()/SS.Combine()} does not guarantee Binding: first, not every set of $F+1$ honest validators may hold shares from the \Disperse{} phase. Second, retrieving shares from different sets of $F+1$ validators may result in different outcome from \fino{SS.Combine()}.   

VSS schemes support Binding through a share-verification functionality \fino{VSS.Verification(share, commit value)}. 
Verification allows validators to check that shares are consistent with some committed value $S’$, 
such that any set of $F+1$ shares applied to \fino{SS.Combine()} results in output $S’$. 
When a validator receives a share, it should verify the share before acknowledging it. 
\Disperse{} completes when there are $N-F$ acknowledgements of verifiable shares, 
certifying that valid shares are held by $F+1$ honest validators and guaranteeing a unique reconstruction by any subset of $F+1$. 
It is possible to add a share-recovery functionality to allow a validator to obtain its individual share prior to \Retrieve{}. 

However, despite the vast progress in VSS schemes (see \autoref{sec:related:vss}), 
share recovery requires linear communication.
Additionally, implementing VSS requires non-trivial cryptography.

% -----*-------*---------*--------*---------------
\subsection*{AVID-M}
\label{sec:otr:avid}
VSS provides stronger guarantees than necessary to satisfy the Binding requirement of \OTR. 
Specifically, \Disperse{} only needs to guarantee a unique outcome from \Retrieve{}, not success. 
Borrowing a technique called AVID-M, which was introduced in DispersedLedger~\cite{yang2022dispersedledger} for reliable information dispersal,
validators can completely avoid verifying shares during \Disperse{}.
Instead, they verify during \Retrieve{} that the (entire) sharing would have a unique outcome, which is very cheap to do. This works as follows. 

To \Disperse{tx}, a user employs \fino{SS.Split(tx-key)} to send individual shares
to validators. This operation requires a deterministic way of
splitting secret shares which, in particular, can be based on validator
identifiers (e.g., $s_i=P(id_i)$, where $P$ is the polynomial hiding \fino{tx-key}).
Additionally, the user combines all shares in a Merkle-tree,
certifies the root, and sends with each share a proof of membership, i.e., a
Merkle tree path to the root.

When a validator receives a share, it should
verify the Merkle tree proof against the certified root before acknowledging the
share. \Disperse{} completes when there are $N-F$ acknowledgements, guaranteeing that
untampered shares are held by $F+1$ honest validators.

During \Retrieve{}, $F+1$  validators reveal individual shares, attaching the
Merkle tree path to prove shares have not been tampered with. Note that $F+1$
shares can be validated by checking only one signature on the root. Validators
use \fino{SS.Combine()} to try to reconstruct \fino{tx-key}.

Then they \emph{post-verify} that every subset of $F+1$ (untampered) shares sent to
validators would have the same outcome. To do this, a validator does not need to
wait for missing shares nor try combinations of $F+1$ shares. Moreover, it does not
need to communicate with others. A validator simply generates missing shares and
re-encodes the Merkle tree. Then, it compares the re-generated Merkle tree with
the root certified by the sender. If the comparison fails, the dealer was faulty
and the validator rejects \fino{tx}. The unique reconstruction is guaranteed because the post-verification
outcome – succeed or fail – becomes fixed upon Disperse completion. Each
validator arrives at the same outcome no matter which subset of $F+1$ (untampered)
shares are input.

AVID-M takes roughly $50 \mu s$ to \fino{SS.split()} and generate Merkle tree commitments,
and roughly $50 \mu s$ for \Retrieve{}, including post-verification, in a 6-out-of-16 scheme.

% -----*-------*---------*--------*---------------
\subsection*{Hybrid}

Even after \Disperse{} has completed, both VSS and AVID-M may need to interact with a specific set of $F+1$ honest validators during \Retrieve{}. 
This implies that the latency of \Retrieve{} is impacted by this specific set of $F+1$ share holders, and not the speed of the fastest 
$F+1$ honest validators. Employing threshold cryptography removes this dependency but incurs a costly computation. 

A Hybrid approach combines the benefits of threshold cryptography with AVID-M. In Hybrid, \Retrieve{} works with (fast) secret-sharing during steady-state, and falls back to threshold cryptography, avoiding waiting for specific $F+1$ validators, during a period of network instability. 

\Disperse{tx-key} is implemented in two parts. 
First, a user applies AVID-M to send validators individual shares and proofs. Second, as a fallback, it sends validators 
\fino{TE.Enc(tx-key)}. 
Once \fino{tx} is committed to the total ordering, \Retrieve{tx-key} has a fast track and a slow track. 
In the fast track, every validator that holds an AVID-M share reveals it. 
A validator that doesn’t hold an AVID-M share for \fino{tx-key} reveals a threshold key decryption share. 
In the slow track, validators may give up on waiting for AVID-M shares and reveal threshold key decryption shares, even if they already revealed AVID-M shares.

Post-verification after $F+1$ valid shares of either kind are revealed checks that \textbf{both} the AVID-M and threshold encryption were shared correctly and would have the same outcome.
More specifically, a validator both re-encrypts 
\fino{tx-key} and re-encodes the AVID-M Merkle tree after reconstruction. 
It compares both outcomes with the sender’s. If the comparison fails, the dealer was faulty and the validator rejects \fino{tx}. 
Binding holds because the post-verification outcome – succeed or fail – becomes fixed upon \Disperse{} completion. 
Each validator arrives at the same outcome no matter which scheme and what subset of $F+1$ (untampered) shares are retrieved, 
thus ensuring Binding.

\section{DAG transport}
\label{sec:dag}
% -----*-------*---------*--------*---------------

In a DAG-based BFT protocol, validators store messages delivered via reliable and causally ordered broadcast in a local graph. A message inserted into the local DAG has the following guarantees: 
%
% \begin{compactdesc}
%
        \textit{Reliability}:
there are copies of the message stored on sufficiently many participants, such that eventually, all honest validators can download it.
        \textit{Non-equivocation}:
messages by each validator are numbered. If a validator delivers some message as the k’th from a particular sender, 
then the message is authenticated by its sender and other validators deliver the same message as the sender k’th message. 
        \textit{Causal Ordering}:
the message carries explicit references to messages which the sender has previously delivered (including its own). Predecessors are delivered locally before the message. % is delivered.

% \end{compactdesc}

Note that the DAGs at different validators may be slightly different at any moment in time. This is inevitable in a distributed system due to message scheduling. However, a DAG-based Consensus protocol allows each participant to interpret its local DAG, reaching an autonomous conclusion that forms total ordering. Reliability, Non-equivocation and Causal Ordering make the design of such protocols extremely simple as we shall see below.

{\bf DAG API.}
A DAG transport exposes two basic API’s, \fino{broadcast()} and \fino{deliver().} 

\begin{compactdesc}

\item[Broadcast.]
\fino{broadcast()} is an asynchronous API for a validator to present payload for the DAG transport to be transmitted to all other validators. The main use of DAG messages in a DAG-based Consensus protocol is to pack meta-information on transactions into a block and present the block for broadcast. The DAG transport adds references to previously delivered messages, and broadcasts a message carrying the block and the causal references to all validators.

\item[Deliver.]
When another validator receives the message carrying the block, it checks whether it needs to retrieve a copy of any of the transactions in the block and in causally preceding messages. Once it obtains a proof-of-availability of all transactions in the block and its causal past, it can acknowledge it. A validator’s upcall \fino{deliver(m)} is triggered when sufficiently many acknowledgments for it are gathered, guaranteeing that the message itself, the transactions it refers to, and its entire causal past maintain Reliability, Non-equivocation and Causal Ordering. 

It is worth noting that in the protocol discussed in this paper, 
the Consensus protocol injects meta-information (e.g., complaints), DAG broadcast never stalls or waits for input from it.

\end{compactdesc}

{\bf Implementing a DAG.}
There are various ways to implement reliable, non-equivocating and causally-ordered broadcast among $N=3F+1 validators$, at most $F$ of which are presumed Byzantine faulty and the rest are honest.

\begin{compactdesc}

\item[Lifetime of a message.]
A validator packs transaction information and meta-information into a message,
adds references to previously delivered messages (including the sender's own preceding messages),
and broadcasts the message to all validators. 

\item[Echoing.] 
The key mechanism for reliability and non-equivocation is for validators to echo a digest of the first message they receive from a sender with a particular index. When $2F+1$ echoes are collected, the message can be delivered. There are two ways to echo, one is all-to-all broadcast over authenticated point-to-point channels a la Bracha Broadcast~\cite{bracha1987asynchronous}; 
the other is converge-cast with cryptographic signatures a la Rampart~\cite{reiter1994securely} and Cachin et al.~\cite{cachin2001secure}.
In either case, echoing can be streamlined so the amortized per-message communication is linear, which is anyway the minimum necessary to spread the message.

\item[Layering.]
Transports are often constructed in a layer-by-layer regime. In this regime, each sender is allowed one message per layer, and a message may refer only to messages in the layer preceding it. Layering is done so as to regulate transmissions and saturate network capacity, and has been demonstrated to be highly effective by various projects~\cite{danezis2018blockmania,gkagol2018aleph,danezis2022narwhal}.
We reiterate that Fino does not require a layered structure. 

\end{compactdesc}

\section{Fino}
\label{sec:fino}
% MEV Protection on a DAG
% --------* ---------- * ---------- * ----------- * ------------

Fino incorporates MEV protection into a \BFT protocol for the partial synchrony model, riding on a DAG transport.
\BFT validators periodically pack pending encrypted transactions into a batch and use the DAG transport to broadcast them.
The key insight for operating \OTR on a DAG is that
each view must wait until \Retrieve{} completes on transactions \fino{tx} of the previous view.

% --- * ----- * ------- * --------
\subsection{The Protocol}

{\bf Views.} The protocol operates in a view-by-view manner. 
Each view is numbered, as in \fino{view(r)}, and has a designated leader known to everyone.

{\bf View change.}
A validator enters view(r+1) when two conditions are met, \viewA and \viewB:
\viewA is satisfied when the local DAG of a validator contains $F+1$ valid votes on 
\fino{proposal(r)} or $2F+1$ valid \fino{complaint(r)} on \fino{view(r)}. 
\viewB is satisfied when every committed transaction \fino{tx} in the local DAG has $F+1$ valid shares revealed, hence \Retrieve{tx} can be completed.
Note that \viewB prevents \BFT validators that do not reveal (correct) shares from enabling the protocol to make progress without opening committed transactions.

{\bf Proposing.} 
When a leader enters a new view(r), it broadcasts \fino{proposal(r)}.\footnote{ 
Recall, broadcast() merely presents payload to be transmitted as scheduled by
the DAG transport, e.g., piggybacked on messages carrying other transaction info.}
Implicitly, \fino{proposal(r)} suggests to commit to the global ordering of
transactions all the messages in the causal history of the proposal.
A leader's \fino{proposal(r)} is {\it valid} if it is well-formatted and is justified in entering \fino{view(r)}.

{\bf Voting.} When a validator sees a valid leader proposal, it broadcasts \fino{vote(r)}.
A validator's \fino{vote(r)} is {\it valid} if it follows a valid \fino{proposal(r)}. 

{\bf Committing.} Whenever a leader's \fino{proposal(r)} has $F+1$ valid votes in the local DAG, 
the proposal and its causal history become {\it committed}.
The commit order induced by a commit decision is described below.

{\bf Share revealing.}
When a validator observes that a transaction \fino{tx} becomes committed, 
it starts \Retrieve{tx} and calls \fino{broadcast()} to present its share of the decryption key \fino{tx-key}.

{\bf Complaining.}
If a validator gives up waiting for a commit to happen in \fino{view(r)}, 
it broadcasts \fino{complaint(r)}.
Note, a \fino{vote(r)} by a validator that causally follows a \fino{complaint(r)} by the validator, if exists, is {\bf not} interpreted as a valid vote.

{\bf Ordering Commits.}
When a validator observes that a leader's \fino{proposal(r)} becomes committed, it orders newly committed transactions as follows:

\begin{compactenum}
\item Let \fino{r'} be the highest view \fino{r' < r} for which \fino{proposal(r')} is in the causal history of 
    \fino{proposal(r)}. \fino{proposal(r')} is recursively ordered.
\item The remaining causal predecessors of \fino{proposal(r)} which have not yet been ordered are appended to the committed sequence
     (within this batch, ordering can be done using any deterministic rule to linearize the partial ordering into a total ordering.)
\end{compactenum}

{\bf Opening Transactions.}
When a validator observes that a leader's \fino{proposal(r)} becomes committed, 
it decrypts every committed transaction \fino{tx} in its causal past 
that hasn't been decrypted already.
That is, let \fino{c} be the highest view \fino{c < r} for which \fino{proposal(r)} causally follows $F+1$ valid votes.
Transactions in the causal past of \fino{proposal(c)} are opened as follows:

%By the view-change rule (above), every transaction that causally precedes \fino{proposal(c)} 
%must have 2F+1 valid shares revealed and will become opened. 
%However, \fino{tx} could be from views earlier or later than \fino{view(c)}, 
%A unique decryption of \fino{tx} is determined as follows. 

\begin{compactenum}
\item Let \fino{r'} be the highest view \fino{r' < c} for which \fino{proposal(r')} is in the causal history of 
    \fino{proposal(c)}. \fino{proposal(r')} is recursively opened.
\item The remaining (committed) transactions \fino{tx} in \fino{proposal(c)}'s causal past 
    are opened using \Retrieve{tx} to either produce a key \fino{tx-key} that
    decrypts \fino{tx} or to reject \fino{tx}. 
\end{compactenum}

% -----*-------*---------*--------*---------------
\begin{figure*}[tb!]
  \centering
  \includegraphics[width=1\linewidth]{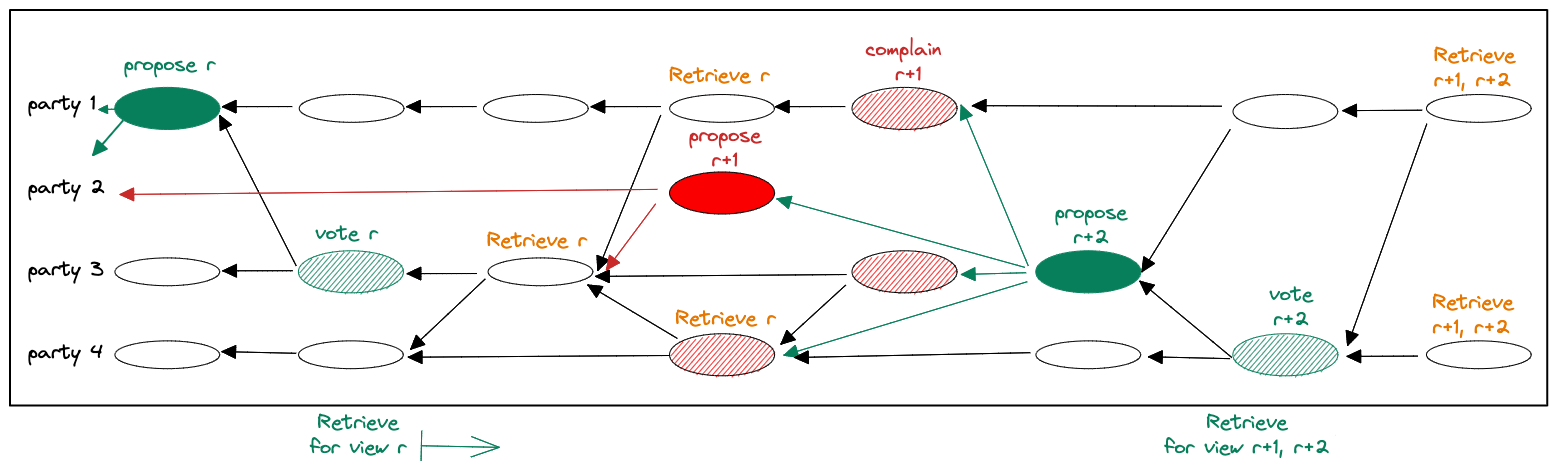}
    \caption{A commit of \fino{proposal(r)} is followed by share retrieval. A commit of \fino{proposal(r+2)} causes an indirect commit of \fino{proposal(r+1)}, followed by share revealing of both.}
  \label{fig:SS2}
\end{figure*}

\noindent
\autoref{fig:SS2} illustrates a couple of Fino scenarios.

{\bf Happy-path scenario.}
In the first view (view(r)), \fino{proposal(r)} becomes committed. 
The commit sets an ordering for transactions in the causal past of \fino{proposal(r)}, 
enabling retrieval of shares for transactions from \fino{proposal(r)}.

%\begin{figure}[bt!]
%  \centering
%  \includegraphics[width=.5\linewidth]{fig/share-expose-1.png}
%    \caption{Commit of \fino{proposal(r)} followed by share retrieval.}
%  \label{fig:SS1}
%\end{figure}

{\bf Scenario with a slow leader.}
A slightly more complex scenario occurs when a view expires because
validators do not observe a leader's proposal becoming committed and they broadcast complaints.
\autoref{fig:SS2} depicts this happening in view(r+1).
Entering view(r+2) is enabled by $2F+1$ complaints about view(r+1).
When \fino{proposal(r+1)} itself becomes committed, it indirectly commits \fino{proposal(r+1)} as well.
Thereafter, 
validators reveal shares for all pending committed transactions, namely, those in both \fino{proposal(r+1)} and \fino{proposal(r+2)}.

\section{Analysis}
\label{sec:analysis}
Fino is minimally integrated with a DAG transport. 
\BFT \OTR logic is embedded into the DAG structure 
simply by broadcasting (more precisely, injecting payloads into broadcasts) in the form of protocol proposals/votes/complaints and revealing shares. 
Importantly, 
at no time is DAG broadcast slowed down by the Fino protocol. 
The reliability and causality properties of the DAG transport make arguing about safety and liveness relatively easy. 

\subsection{Safety}

\begin{lemma} \label{lem:vote}
Assume a \fino{proposal(r)} ever becomes committed.
Denote by $r'$ the minimal view, where $r' > r$, such that \fino{proposal(r')} ever becomes committed.
Then for every $r \leq q < r'$, \fino{proposal(q+1)} causally follows \fino{proposal(r)}.
\end{lemma}

\begin{proof} (Sketch)
Since \fino{proposal(r)} becomes committed, by definition $F+1$ parties sent votes for it.
There are two possibilities for a leader's \fino{proposal(q+1)} to be valid. 
The first is \fino{proposal(q+1)} may reference $F+1$ valid \fino{vote(q)} messages.
This case occurs only for $q == r$ by the lemma assumption that $r' \geq q+1$ is the minimal view that ever becomes committed.
The second possibility can occur
for all $r \leq q < r'$, namely, \fino{proposal(q+1)} references $2F+1$ valid \fino{complaint(q)} messages.

In the first case, 
valid \fino{vote(r)} causally follows \fino{proposal(r)}, 
and a fortiori \fino{proposal(r+1)} causally follows \fino{proposal(r)}. 

In the second case, one of $2F+1$ \fino{complaint(q)} messages is sent by a party who sent a valid \fino{vote(r)}. 
By definition, \fino{vote(r)} must precede \fino{complaint(q)},
otherwise it is not considered a (valid) vote.
Hence, \fino{proposal(q+1)} causally follows \fino{complaint(q)} which follows \fino{proposal(r)}. 
\end{proof}

\begin{lemma} \label{lem:vote2}
If ever a \fino{proposal(r)} becomes committed, 
then every valid \fino{proposal(q)}, where \fino{q > r}, causally follows \fino{proposal(r)}.
\end{lemma}

\begin{proof} (Sketch)
The proof follows by a simple induction on Lemma~\ref{lem:vote}.
\end{proof}

\begin{lemma} \label{lem:agreement}
If an honest party commits \fino{proposal(r)},
and another honest party commits \fino{proposal(r')}, where $r' > r$, 
then the sequence of transactions committed by \fino{proposal(r)}
is a prefix of the sequence committed by \fino{proposal(r')}.
\end{lemma}

\begin{proof} (Sketch)
When \fino{proposal(r')} becomes committed, 
the commit ordering rule is recursively applied to valid proposals in its causal past.
By Lemma~\ref{lem:vote2}, the (committed) \fino{proposal(r)} is a causal predecessor of every valid \fino{proposal(s)},
 for $r < s \leq r'$, and eventually the recursion gets to it.
\end{proof}

\begin{lemma} \label{len:unique}
% Unique opening.
    If ever a \fino{proposal(r)} becomes committed, then every committed transaction
    proposed by a valid \fino{proposal(r')}, where \fino{r' < r}, can be
    uniquely decrypted.
\end{lemma}

\begin{proof} (Sketch)
    When \fino{proposal(r)} becomes committed, then by
    Lemma~\ref{lem:agreement}, all the transactions from previously committed
    proposals are part of the committed transaction history.  Moreover,
    committing \fino{proposal(r)} implies the view change (from \fino{r-1})
    which occurs only if the condition viewB (see \autoref{sec:fino}) is met,
    i.e., every committed transaction \fino{tx} in the local DAG can complete \fino{Retrieve(tx)}, hence
    can be uniquely decrypted (or rejected).

\end{proof}

\subsection{Liveness}

Liveness is guaranteed after GST, 
i.e., after communication has become synchronous with a known
$\Delta$ upper bound on transmission delays.

After GST, 
views are semi-synchronized through the DAG.
In particular, 
suppose that after GST, every broadcast by an honest party arrives at all honest
parties within $DD$.
Once a \fino{view(r)} with an honest leader is entered by the first honest party, 
within $DD$ all the messages seen by one party are delivered by both the leader and all other honest parties. 
Hence, within $DD$, all honest parties enter \fino{view(r)} as well. 
Within two additional DAG latencies, $2 \cdot DD$, the \fino{view(r)} proposal and votes from all honest parties are spread to everyone. 

Assuming view timers are set to be at least $3 \cdot DD$, 
once \fino{view(r)} is entered, 
a future view will not interrupt a commit. 
In order to start a future view, 
its leader must collect either $F+1$ \fino{vote(r)} messages, hence commit \fino{proposal(r)}; or $2F+1$ \fino{complaint(r)} expiration messages, which is impossible as argued above. 

% As for satisfying \fino{Reconstruct(tx)} for committed transactions \fino{tx}, 
% this is guaranteed, in the worst case, through threshold decryption.

\subsection{Communication complexity} 
Protocols for the partial synchrony model have unbounded worst case by nature, hence, we concentrate on the costs incurred during steady state when a leader is honest and communication with it is synchronous:

{\bf DAG message cost.} In order for DAG messages to be delivered reliably, it must implement reliable broadcast.
This incurs either a quadratic number of messages carried over authenticated channels, or a quadratic number of signature verifications, per broadcast. 
In either case, the quadratic cost may be amortized by pipelining, driving it in practice to (almost) linear per message.

{\bf Commit message cost.} Fino sends $F+1$ broadcast messages, a proposal and votes, per decision. 
A decision commits the causal history of the proposal, consisting of (at least) a linear number of messages. 
Moreover, each message may carry multiple transactions in its payload.
As a result, in practice the commit cost is amortized over many transactions.

\subsection{Latency}

{\bf Commit latency.} The commit latency in terms of DAG messages is 2, one proposal followed by votes.

{\bf Opening  latency.}
During periods of stability, there are no complaints about honest leaders by any honest party. 
If \fino{tx} is proposed by an honest leader in \fino{view(r)}, it will receive $F+1$ votes and become committed within two DAG latencies.
Within one more DAG latency, 
every honest party will post a message containing a share for \fino{tx}. 
Thereafter, whenever $F+1$ are available, everyone will be able to complete \Retrieve{tx} with a unique outcome. 

\subsection{Faulty User}
One drawback of the proposed protocol is that it assumes a well-connected and
honest user (dealer), able to entrust $F+1$ honest validators with their shares.
If the dealer is malicious or unable to broadcast shares it can lead to the
situation where the transport layer buffers a transaction indefinitely without being able to \fino{deliver()} it. 
The Hybrid approach we presented in \autoref{sec:framework} solves
this issue by using threshold encryption 
as a fallback mechanisms (i.e., whenever enough shares cannot be received,
parties would proceed with threshold decryption).
However, the hybrid construction suffers the slowness of threshold encryption to verify a unique outcome even in faultless executions.

\section{Evaluation (prelim)}
\label{sec:implementation}
% \subsection{Microbenchmarks of Transaction Blinding}

We implemented the proposed variants of the Fino Disperse and Retrieve
functionalities
and investigated the computational
overhead that these schemes introduce.
For secret sharing we implemented SSS~\cite{shamir1979share}, while for threshold encryption we implemented the
TDH2 scheme~\cite{shoup1998securing}.   We selected the
schemes with the most efficient cryptographic primitives we had access to, i.e.,
the secret sharing scheme uses the Ed25519
curve~\cite{bernstein2012high}, while TDH2 uses ristretto255~\cite{VGT+19}
as the underlying prime-order group.  Performance for both schemes is presented in a
setting where 6 shares out of 16 are required to recover the plaintext.

\begin{table}[t!]
\ifx\tics\undefined
    \small
\fi
\caption{Comparison of the implementations' performance.}
    \label{tab:perf}
\begin{tabular}{|l|r|r|r|r|}
\hline
Scheme & \texttt{Disperse} & \texttt{ShareGen} & \texttt{ShareVerify} &
    \texttt{Reconstruct} \\ \hline\hline
TDH2-based & 311.6$\mu s$ & 434.8$\mu s$ & 492.5$\mu s$ & 763.9$\mu s$  \\ \hline
SSS-AVID-M& 52.7$\mu s$ &   N/A   & 2.7$\mu s$ & 52.5$\mu s$  \\ \hline
Hybrid opt. & 360.3$\mu s$ &   N/A   & 2.7$\mu s$ & 361.5$\mu s$  \\ \hline
Hybrid pess. & 360.3$\mu s$ &   434.8$\mu s$   & 492.5$\mu s$ & 828.9$\mu s$  \\ \hline
\end{tabular}
\end{table}

The results presented in \autoref{tab:perf} are obtained on an Apple M1 Pro.
\texttt{Disperse}
refers to the overhead on the client-side while \texttt{ShareGen} is the operation of
deriving a decryption share from the TDH2 ciphertext (it is absent
in SSS and the optimistic path of the hybrid scheme).  In TDH2,
\texttt{ShareVerify} verifies if
a decryption share matches the ciphertext, while in SSS-based Fino it only
checks whether a share belongs to the tree aggregated by the signed Merkle root
attached by the client. \texttt{Reconstruct} recovers the plaintext from the ciphertext and
the number of shares and verifies the outcome integrity (i.e., AVID-M and
threshold re-encryption
verification, where applicable).

As demonstrated by these micro-measurements, the SSS-AVID-M scheme is the most
efficient:  in our blind ordering scenario using threshold encryption, each party processing a TDH2 ciphertext
would call \texttt{ShareGen} once to derive its decryption share,
\texttt{ShareVerify} \textit{k-1} times
to verify the threshold number of received shares, and \texttt{Reconstruct} once to obtain
the plaintext.  Assuming \textit{k=6}, the total computational overhead for a single
transaction to be recovered would take around 3.7ms CPU time. % This constitutes the CPU overhead
% required in the pessimistic path for a transaction.
%
With secret-sharing based Fino, the party would also call \texttt{ShareVerify} \textit{k-1} times and
\texttt{Reconstruct} once, which
requires only 66$\mu s$ CPU time.
The hybrid approach introduces an interesting trade-off between its properties
and performance. In the optimistic path, it requires around 378$\mu s$ CPU time
per transaction, while its fallback introduces negligible overhead to the plain
threshold-encryption scheme.

% The last operation that we evaluated is the verification, ensuring that the
% secret shares and threshold ciphertext produce the same output. In our setting,
% the verification in the optimistic case takes 338$\mu s$ CPU time on average.
% That means, that in the optimistic path  the total overhead would be around
% 355$\mu s$ per transaction, which is an order of magnitude faster than with the
% threshold encryption.  In the pessimistic case, verification requires 49$\mu
% s$ of CPU time on average, only slightly increasing the total computation time
% of the pessimistic path.
%
We emphasize that these are micro-benchmarks of the blinding/unblinding of transactions;
in the future, we plan to complete a performance evaluation of the entire DAG-based blind-ordering protocol.

% \subsection{Other Considerations}
% \psz{dropping it, since we use these two methods now}
%
% Besides the higher performance overhead, TDH2 requires a trusted setup, but the
% scheme also provides some advantages over secret sharing. For instance, a TDH2
% ciphertext may be sent only to a single party and the network will be able to recover the
% plaintext. An SSS-based requires the client to send shares directly to multiple parties.
% Meanwhile,
% waiting for the network to receive shares for a transaction,
% the transaction occupies buffer space on parties' machines and would need
% to be either expired by the parties (possibly violating liveness) or kept in the
% state forever (possibly introducing a denial-of-service vector).  Moreover, SSS
% requires a trusted channel between clients and parties, which is not required by
% TDH2 itself.  Finally, in lieu of secret-share verification in Fino,
% setting a deterministic subset of shares used for decrypting a transaction requires a consensus decision.

\section{Related Work and Discussion}
\label{sec:related}
Despite the MEV problem being relatively new, there already exist
consensus-related systems aiming at solving it. For instance, Flash Freezing Flash
Boys (F3B)~\cite{zhangflash} is a commit-and-reveal architecture, where clients
send TDH2-encrypted ciphertexts to be ordered by Consensus, and afterward to be
decrypted via shares released by a dedicated secret-management committee.
Other closely related techniques and schemes helpful in mitigating MEV are discussed below.

\subsection{Verifiable Secret Sharing (VSS)}
\label{sec:related:vss}

Although we ultimately end up forgoing verifiability/recoverability of shares, 
VSS is used in many settings like ours. The overall communication complexity
incurred in VSS on the dealer sharing a secret and on a party recovering a share
has dramatically improved in recent years.
VSS can be implemented inside the asynchronous echo broadcast protocol in $O(n^3)$
communication complexity using Pederson's original two-dimensional polynomial scheme
Non-Interactive Polynomial Commitments~\cite{pedersen1991non}.
Kate et al.~\cite{kate2010constant} introduce a VSS scheme with $O(n^2)$ communication complexity,
utilized for asynchronous VSS by Backes et al.~\cite{backes2013asynchronous}.
Basu et al.~\cite{basu2019efficient} propose a scheme with linear $O(n)$ communication complexity.

A related notion is robust secret sharing (RSS) introduced by
Krawczyk~\cite{krawczyk1993secret} and later revised by Bellare and
Rogaway~\cite{bellare2007robust}. RSS allows recoverability in case of
incorrect (not just missing) shares, however, these schemes assume an honest
dealer.  Duan et al.~\cite{duan2017secure} propose ARSS extending RSS to the
asynchronous setting.

\subsection{Time-Based Order-Fairness}

Blind Order-Fairness is achieved by deterministically ordering encrypted
transactions and then, after the order is final, decrypting them.
The deterministic order can be enhanced by sophisticated ordering logic
present in other protocols.
In particular, Fino can be extended to provide Time-Based Fairness
additionally ensuring that the received transactions are not only unreadable by
parties, but also their relative order cannot be influenced by malicious parties
(the transaction order would be defined by the time of its ciphertext arrival).

For instance, Pompē~\cite{zhang2020byzantine} proposes a property called {\em Ordering Linearizability}:
{\it ``if all correct parties timestamp transactions tx, tx' such that tx' has a
lower timestamp than tx by everyone, then tx' is ordered before tx.''}
It implements the property based on an observation that if parties
exchange transactions associated with their receiving timestamps, then for each
transaction its median timestamp, computed out of 2F+1 timestamps collected, is
between the minimum and maximum timestamps of honest parties.
Fino can be easily extended by the Linearizability property offered by Pompē and the final
protocol is similar to the Fino with Blind Order-Fairness (see above) with
only one modification. Namely,
every time a new batch of transactions becomes committed,
parties independently sort transactions by their aggregate (median) timestamps.

More generally, 
Fino can easily incorporate other Time-based Fairness ordering logic.
Note that in Fino, the ordering of transactions is determined on
encrypted transactions, but time ordering information should be open. 
The share revealing, share collection, and unique decryption following a committed ordering are the same as
presented previously. The final protocol offers much stronger properties
since it not only hides payloads of unordered transactions from parties,
but also prevents parties from reordering received transactions.

One form of Time-based Order-Fairness introduced in Aequitas~\cite{kelkar2020order} 
is {\em Batch-Order Fairness}:
{\it ``if sufficiently many (at least ½ of the) parties receive a transaction tx before
another transaction tx', then no honest party can deliver tx in a block after tx', ''}
and a protocol achieving it.
Other forms of Time-Based Order-Fairness which may be used in Fino include
``Wendy Grows Up''~\cite{kursawe2021wendy}, which introduced {\em Timed Relative Fairness},
{\it ``if there is a time t such that all honest parties saw (according to their
local clock) tx before t and tx' after t , then tx must be scheduled before tx', ''}
and 
``Quick Order Fairness''~\cite{cachin2021quick},
which defined {\em Differential-Order Fairness}, 
{\it ``when the number of correct parties that broadcast tx before tx' exceeds the
number that broadcast tx' before tx by more than $2F + \kappa$, for some $\kappa \geq 0$, then
the protocol must not deliver tx' before tx (but they may be delivered together).''}

Themis~\cite{kelkar2021themis} is a 
protocol realizing {\em Batch-Order Fairness}, where parties do not rely on timestamps
(as in Pompē) but only on their relative transaction orders reported.  Themis
can also be integrated with Fino, however, to make it compatible this design
requires some modifications to Fin's underlying DAG protocol. More
concretely, Themis assumes that the fraction of bad parties cannot be one quarter, i.e., F out of 4F+1. 
A leader makes a proposal based on 3F+1 out of 4F+1 transaction orderings (each
reported by a distinct party). Therefore, we would need to modify the DAG transport
so that parties reference 3F+1 preceding messages (rather than 2F+1).

% -----*-------*---------*--------*---------------
\subsection{DAG-based BFT}
There are many known \BFT solutions for partial synchrony, 
and more specifically, 
several recent solutions that ride on a DAG~\cite{danezis2022narwhal,giridharan2022bullshark,keidar2022cordial}.
When constructing Fino, we wanted to build MEV protection into a simple DAG-based \BFT solution, described below.
Notwithstanding, we remark that Fino's MEV protection can possibly be incorporated into other DAG-riding \BFT solutions.

We borrowed a subprotocol from Bullshark~\cite{giridharan2022bullshark} that deals with partial synchrony,
modifying it so that DAG transmissions would never stall waiting for \BFT protocol steps or timers.
Fino is different from the borrowed Bullshark component in that \BFT protocol 
steps (e.g., view changes, proposals, votes, and complaints) are injected into the DAG at any time, independent of DAG layers.
It uses \fino{broadcast()} (defined in \autoref{sec:pre}) to present these steps as payloads to the DAG transport, completely asynchronously,
while normal DAG transmissions continue.
A hallmark of the DAG-riding approach, which Fino preserves while adding Blind Order-Fairness, is zero message overhead.

Narwhal~\cite{danezis2022narwhal} is a recent DAG transport that has a layer-by-layer structure, each layer having at most one message per sender and referring to 2F+1 messages in the preceding layer. A similarly layered DAG construction appears earlier in
Aleph~\cite{gkagol2018aleph},
but it does not underscore the separation of transaction dissemination from DAG messages.
Narwhal-HS is a BFT Consensus protocol within~\cite{danezis2022narwhal} based on
HotStuff~\cite{yin2019hotstuff} for the partial synchrony model,
in which Narwhal is used as a "mempool". 
In order to drive Consensus decisions, 
Narwhal-HS adds messages outside Narwhal, 
using the DAG only for spreading transactions.

DAG-Rider~\cite{keidar2021all} and
Tusk~\cite{danezis2022narwhal}
build randomized BFT Consensus for the asynchronous model riding on Narwhal, 
These protocols are zero message overhead over the DAG, not exchanging any messages outside the Narwhal protocol.
Both are structured with purpose-built DAG layers grouped into "waves" of 4 (2) layers each. 
Narwal waits for the Consensus protocol to inject input value every wave, though in practice, this does not delay the DAG materially. 

Bullshark~\cite{giridharan2022bullshark} builds BFT Consensus riding on Narwhal for the partial synchrony model.
It is designed with 8-layer waves driving commit, each layer purpose-built to serve a different step in the protocol.
Bullshark is a "zero message overhead" protocol over the DAG, however, 
due to a rigid wave-by-wave structure, 
the DAG is modified to wait for Bullshark timers/steps to insert transactions into the DAG.
In particular, if leader(s) of a wave are faulty or slow, some DAG layers wait to fill until consensus timers expire.

More generally, since the eighties, causally ordered reliable broadcast has been utilized in forming 
distributed Consensus protocols, e.g.,~\cite{birman1987exploiting,peterson1989preserving,melliar1990broadcast,moser1999byzantine,amir1991transis,dolev1993early}. 

The notion of the secure, causal, reliable broadcast was introduced by Reiter and Birman~\cite{reiter1994securely}, and later refined by Cachin~\cite{cachin2001secure} and Duan et al.~\cite{duan2017secure}.
This primitive was utilized in a variety of BFT replicated systems, 
but not necessarily in the form of zero message overhead protocols riding on a DAG. 
Several of these pre blockchain-era BFT protocols are DAG based, 
notably Total~\cite{moser1999byzantine} and ToTo~\cite{dolev1993early}, both of which are BFT solutions for the asynchronous model.

\section{Conclusions}
\label{sec:conclusions}
In this paper, we show how DAG-based BFT protocols can be enhanced
to mitigate MEV -- arguably one of the main threats to the success of cryptocurrencies.
We investigate the design space of achieving Blind Order-Fairness, and we
propose an approach that focuses on practicality. We present preliminary correctness
and performance results which indicate that our scheme is secure and efficient.
In the future, we plan to extend our analysis and experiments, and investigate 
other cryptographic tools which can improve our scheme.

\ifx\tics\undefined
    \section*{Acknowledgment}
    We are grateful to Soumya Basu, Christian Cachin, Ari Juels, Mahimna Kelkar, Lefteris Kokoris-Kogias, Oded Naor, Mike Reiter and Nibesh Shrestha for many comments that helped improve this writeup.
\fi

\bibliography{ref}

\ifx\tics
    \newpage
\fi
%\appendix

% by now, this would need updating; not sure it's worth it
%\newpage
%\section{Fino Pseudo-Code}
%\label{app:fino}
%\input{sec/pcode}

\end{document}